\documentclass[conference]{IEEEtran}
\makeatletter
\def\ps@headings{%
\def\@oddhead{\mbox{}\scriptsize\rightmark \hfil \thepage}%
\def\@evenhead{\scriptsize\thepage \hfil \leftmark\mbox{}}%
\def\@oddfoot{}%
\def\@evenfoot{}}
\makeatother
\ifCLASSINFOpdf
  \usepackage[pdftex]{graphicx}

\else
 \usepackage[dvips]{graphicx}

\fi

\usepackage{cite}
\usepackage{times}
\usepackage{dsfont}
\usepackage{cite}
\usepackage{times}

\usepackage{graphicx}
\usepackage{subfigure}
\usepackage{color}
\usepackage{ifpdf}
\usepackage{epsfig}
\usepackage{latexsym}
\usepackage{amsfonts}
\usepackage{amsmath}
\usepackage{amssymb}
\usepackage{paralist}
\usepackage{comment}
\usepackage{xspace}
\usepackage{mathrsfs}
\usepackage{amssymb}
\usepackage{slashbox}
\usepackage{color}
\usepackage[small]{caption2}
\usepackage[ruled]{algorithm}
\usepackage[noend]{algorithmic}

\usepackage{amssymb}
\usepackage{url,epsfig,array}
\usepackage{leftidx}
\usepackage{amsmath}
\usepackage{multirow}
\usepackage{url}

\def\ie{\textit{i.e.}\xspace}
\def\etal{\textit{et al.}\xspace}

\def\eg{\textit{e.g.}\xspace}

\def\aka{\textit{a.k.a.}\xspace}

\newtheorem{theorem}{Theorem}

\newtheorem{lemma}[theorem]{Lemma}

\newcommand{\LXYCUT}[1]{{}}

\begin{document}
\bibliographystyle{plain}

%
\title{PPS: Privacy-Preserving Strategyproof Social-Efficient Spectrum Auction Mechanisms}
\author{\IEEEauthorblockN{He Huang\IEEEauthorrefmark{1}, Xiang-Yang Li\IEEEauthorrefmark{3}\IEEEauthorrefmark{4}, Yu-e Sun\IEEEauthorrefmark{2}, Hongli Xu\IEEEauthorrefmark{5}, and Liusheng Huang\IEEEauthorrefmark{5}}
\IEEEauthorblockA{\IEEEauthorrefmark{1}School of Computer Science and Technology, Soochow University, Suzhou, China\\
\IEEEauthorrefmark{2}School of Urban Rail Transportation, Soochow University, Suzhou, China\\
\IEEEauthorrefmark{3}Department of Computer Science, Illinois Institute of Technology, Chicago, USA\\
\IEEEauthorrefmark{4}Tsinghua National Laboratory for Information Science and Technology (TNLIST), Tsinghua University, China\\
\IEEEauthorrefmark{5}Department of Computer Science and Technology,
University of Science and Technology of China, Hefei, China
}}

\maketitle

\begin{abstract}
Many spectrum auction mechanisms have been proposed for spectrum
allocation problem, and unfortunately,
few of them protect the bid privacy of bidders and achieve good social
efficiency.
In this paper, we  propose PPS,  a \underline{P}rivacy \underline{P}reserving
\underline{S}trategyproof spectrum auction framework.
Then, we design two schemes based on PPS separately for
1) the Single-Unit Auction model (SUA), where only
single channel to be sold in the spectrum market;
and 2) the Multi-Unit Auction model (MUA),
where the primary user subleases multi-unit
channels to the secondary users and each of the secondary users  wants
to access multi-unit channels either. Since the social
efficiency maximization problem is NP-hard in both auction
models, we present  allocation mechanisms
with approximation factors of $(1+\epsilon)$ and $32$ separately
for SUA and MUA, and further judiciously design strategyproof auction
mechanisms with privacy preserving based on them.
Our extensive evaluations show that our mechanisms
 achieve good social efficiency and with low
computation and communication overhead.
\end{abstract}



\section{Introduction}\label{sec:intro}
The ever-increasing demand for limited radio spectrum resource poses a great
challenge in spectrum allocation and usage \cite{xu2011efficient}.
Recent years, auction has been widely regarded as a preeminent way to tackle such a challenge
 because of its fairness and efficiency \cite{krishna2009auction}. In general, bidders in spectrum auctions are
 the secondary users, while the auctioneer is a primary user in the
 single-sided spectrum auctions.

In recent years, many strategyproof auction mechanisms,
 in which  bidding the true valuation is the dominant strategy of
 bidders, have been proposed for solving spectrum allocation issue.
Unfortunately, the auctioneer is not always trustworthy.
Once the true valuations of bidders are revealed to a corrupt auctioneer,
he may abuse such information
to improve his own advantage.
Besides, the true valuation may divulge the profit of bidders,
which is also a commercial secret for each bidder.
Therefore, bid privacy preservation should be considered in
spectrum auction design.
However, only few studies (\eg,
~\cite{wufan2013spring,pan2011purging}) were proposed
to protect the bid privacy of bidders.

Allocating channels to the buyers who \textbf{value} them most
 will improve the \emph{social efficiency}.
There have been many studies devoted to maximizing the social efficiency
while ensuring strategyproofness in spectrum auction mechanism design ~\cite{gopinathan2010strategyproof,dong2012combinatorial,huang2013mobihoc,zhu2012truthful,xu2011efficient}.
Unfortunately, none of these auction mechanisms provides any guarantee on bid
privacy preservation.

In this paper, we consider the issue of designing strategyproof spectrum auction
mechanism which maximizes the social efficiency while protecting the bid
privacy of bidders.
We propose a \underline{P}rivacy \underline{P}reserving \underline{S}trategyproof spectrum auction framework (PPS).
Under PPS, we mainly study two models: 1) the \underline{S}ingle-\underline{U}nit \underline{A}uction model (SUA) and
2) the \underline{M}ulti-\underline{U}nit \underline{A}uction model (MUA).
In the SUA model, the auction mechanism design only focuses on single channel trading.
Multi-unit channels trading is supported in the case of MUA model.
Since the maximization of social efficiency problem in both SUA and MUA are NP-hard,
we design  allocation mechanisms with approximation factors of ($1+\epsilon$) and
$32$ separately for the SUA and the MUA.
We show that the proposed approximation allocation
mechanisms are bid-monotone, and further design strategyproof auction mechanisms based on them,
which are denoted as PPS-SUA and PPS-MUA respectively.
As the PPS-MUA only ensures the worst case performance,
we further propose an improved mechanism, denoted by PPS-EMUA, to
improve the social efficiency  of PPS-MUA.
We also show that PPS-EMUA is strategyproof and privacy-preserving.

It is not a trivial job to protect privacy  of the true bid
values of bidders in the auction mechanisms as auction relies on these
bid values to make  decision on allocation and payment computation.
Notice that, for maximizing social efficiency and computing payment,
 we need to compute many various
 bid sums of conflict-free bidders in our allocation mechanisms.
However, it is hard for the auctioneer or
 the bidders to compute these bid sums with privacy preserving since
  the auctioneer does not know any bidder's true bid value.
To address these challenges, we will first introduce an \emph{agent},
which is a semi-trusted third party (such as FCC), different from auctioneer.
The agent, together with
the auctioneer, will execute the auction in PPS.
In our design, bidders apply \emph{Paillier's homomorphic
 encryption} to encrypt the bids so agent can perform computation on
 the ciphertexts,
 agent then sends the results by adding random numbers and shuffling
 bidder IDs to auctioneer for making  allocation decision,
 which provides privacy protection without affecting the
 correctness of the allocation.
We will  prove that neither the agent nor the auctioneer can
 infer any true bid value about the bidders without collusion.
To the best of our knowledge, PPS is the first privacy preserving
 spectrum auction scheme that maximizes the social efficiency.
 Note that we did not focus on protecting the location privacy of bidders in our
  mechanisms, as previous schemes (\eg, \cite{taeho-li-info13}) can be integrated into our mechanisms.

The remainder of paper is organized as follows.
In Section \ref{sec:prelim},
we  formulate the spectrum auction  and present the framework of PPS.
Section \ref{sec:homogeneous} proposes a strategyproof spectrum auction
mechanism for solving the single-unit auction model.
Section \ref{sec:heterogeneous} further extends the auction model with consideration of multiple-items trading model.
Extensive simulation results are evaluated in Section \ref{sec:simulation}.
Section \ref{sec:review} discusses the related literatures and section \ref{sec:conclusion} concludes the paper.

\section{Problem Formulation and Preliminaries}\label{sec:prelim}

\subsection{\textbf{Spectrum Auction Model}}\label{sec:model}

We model the procedure of secure spectrum allocation as a sealed-bid auction, in which there is
an \emph{auctioneer} (\aka primary user), a set of \emph{bidders} (\aka secondary users) and an \emph{agent}.
In each round of the auction, the auctioneer subleases the access right of $m$ channels
to $n$ bidders. The bidders first encrypt their bids by using the
\emph{encryption key} of a homomorphic encryption scheme (\eg,
Paillier's scheme) for
the auctioneer, and submit the encrypted bids to the \emph{agent} (not
the auctioneer).
Here, $E(m)$ denotes the homomorphic encryption of message $m$.
Then, the auctioneer and the agent allocate the channels to the bidders via communicating
with each other. We assume that the agent is a \emph{semi-trusted
  party}, and will not collude with the auctioneer.

We use $\mathcal{C}=\{c_1,...,c_m\}$ to denote the set of channels, and $\mathcal{B}=\{1,...,n\}$
to denote the set of bidders. Each bidder $i\in \mathcal{B}$ is described as $i=\{L_i, N_i, b_i, v_i, p_i\}$,
where $L_i$ is the geographical location of $i$,
$N_i$ is the number of channels that bidder $i$ wants to buy,
$b_i$, $v_i$ and $p_i$ separately denote the bid value, true valuation and payment of $i$ for all the
channels that he wants to buy.
We assume that the interference radii of all channels are the same,
which are equal to $\frac{1}{2}$ unit. Then,  two bidders $i$ and $j$ conflict with each other if the distance
between $L_i$ and $L_j$ is smaller $1$ unit. Bidders can share one
channel iff they are conflict free with each other.

In this paper, we study two spectrum auction models. The first one is that
there is only one channel in the spectrum
market, then $m=1$ and $N_i=1$ for each bidder. We call this model the \emph{Single-Unit Auction model} (SUA).
The second one is  the \emph{Multi-Unit Auction model} (MUA) which supports multiple channels trading in the market.
In MUA, each bidder wants to access $N_i \ge 1$ channels rather than
part of them.

\subsection{\textbf{Design Targets}}\label{sec:target}

Our work is to design social efficient strategyproof spectrum
 auction mechanisms with bid privacy preservation.
Firstly, we will allocate channels to the bidders who value them most to
 maximize the social efficiency.
 However, the optimal channel allocation problem  in SUA and MUA
 are all NP-hard. Thus, we will design approximation
 mechanisms  instead.
Secondly, our auction mechanisms should be strategyproof, which means
 bidding truthfully
 is the \emph{dominant strategy} for any bidders.
To achieve this,
 it is sufficient  to show that our allocation mechanism is
 \emph{bid-monotone}, and always charges each winner its
 \emph{critical value} \cite{nisan2007algorithmic}.
We say an allocation mechanism is bid-monotone if
 bidder $i$ wins the auction by bidding $b_i$, he will always
 win by bidding $b'_{i}>b_i$.
And the critical value of each bidder $i$ in a bid-monotone allocation
 mechanism is  the minimum bid that bidder $i$ will win in the auction.
The third objective  is to protect the privacy of the bid values of
bidders.
To achieve privacy protection, we will apply homomorphic encryption to
encrypt the bid values using the public key of auctioneer, and agent
will perform the most of the
computation and send the intermediate results to the auctioneer.
We will show that
both the auctioneer and the agent cannot get any  information about
the true bid values of bidders as long as they will not collude with each other.

\LXYCUT{
\subsection{\textbf{Homomorphic Encryption}}\label{sec:homo}

We use Paillier's homomorphic encryption system in this paper,
which satisfies the following homomorphic
operation:
\begin{equation*}
  E(\textsf{msg}_1)E(\textsf{msg}_2)=E(\textsf{msg}_1+\textsf{msg}_2)
\end{equation*}
\begin{equation*}
  E(\textsf{msg}_1)^{\textsf{msg}_2}=E(\textsf{msg}_1\textsf{msg}_2)
\end{equation*}
where $E(\textsf{msg}_i)$ is the ciphertext of message $\textsf{msg}_i$.
}

\subsection{\textbf{A Spectrum Auction Framework with Privacy Preserving}}\label{sec:framework}

The process of our spectrum auction mechanisms consists of three steps: bidding, allocation
and payment calculation.
To protect the bid values of bidders, we design a strategyproof spectrum auction framework with
privacy preserving, namely PPS, which is shown in Algorithm \ref{alg:0}.

\begin{algorithm}
\caption{PPS: Privacy Preserving Strategyproof Spectrum Auction
  Framework}\label{alg:0}
{
\begin{algorithmic}[1]
\STATE Each bidder $i$ submits $E(b_i)$, $N_i$ and $L_i$ to the agent, where $b_i$ is
encrypted by using the encryption key of the auctioneer;
\STATE The agent and the auctioneer run a bid-monotone allocation mechanism while protecting the
bid privacy of bidders.
\STATE The agent and the auctioneer compute a critical value for each winner with bid privacy
preserving.
\end{algorithmic}}
\end{algorithm}

\section{A Single-unit Scheme}\label{sec:homogeneous}
In this section, we will present a strategyproof spectrum auction mechanism for SUA,
denoted by PPS-SUA, which maximizes the social efficiency
and preserves the bid privacy.

\subsection{\textbf{Initialization and Bidding}}\label{sec:suaini}

Before running the auction, the auctioneer generates an encryption key $EK$
and a decryption key $DK$ of Paillier's cryptosystem. Then, he announces $EK$ as the public key, and keeps $DK$ in private.
Each bidder $i$ encrypts his bid $b_i$ by using $EK$, and sends $(E(b_i), L_i)$ to the agent.
In the sending procedure, each bidder keeps his encrypted bidding price as a secret to the auctioneer.

\subsection{\textbf{Allocation Mechanism with Privacy Preserving}}

After receiving the encrypted bids from bidders, the auctioneer and
the agent allocate channels to bidders via communicating with
each other. The goal of our allocation mechanism is to maximize
the social efficiency, which is equal to finding a group of
conflict-free bidders with highest bid sum, which is a well-known
NP-hard problem.
To tackle this NP-hardness, we propose a
polynomial time approximation scheme (PTAS) based on \emph{shifting strategy} \cite{hunt1998nc,li2006simple},
which provides an approximation factor of ($1+\epsilon$).
For completeness of presentation, we first review this PTAS method.

\begin{figure}[!t]
\centering
\includegraphics[width=2.5in,height=2in]{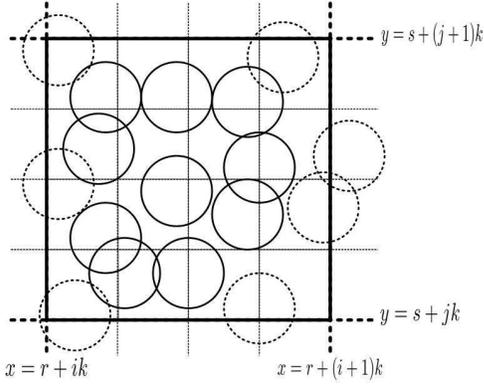}
\caption{A grid subdivided by (r,s)-shifting ($k=4$).}
\label{fig:shifting}
\end{figure}

In the PTAS, we first select a positive integer $k$, then, the plane is subdivided into several grids
of size at $k*k$ by a collection of vertical lines $x=i\cdot k +r$ and horizontal lines $y=j\cdot k+s$,
where $0\le r,s \le k-1$. We call such a subdivision as \emph{($r,s$)-shifting}.
Here we assume that the conflict radius of each bidder
is $\frac{1}{2}$,
then each bidder can be viewed as a \emph{unit disk}.
Fig. \ref{fig:shifting} gives an instance of a grid subdivided by $(r,s)$-shifting, where $k=4$.
We will throw away all the disks which intersect with some
special lines $X \equiv r \mod k$ and $Y \equiv s \mod k$ in ($r,s$)-shifting, and solve the sub-instances of disks
contained in each grid individually. Here, a grid is a square defined by $\{(x, y)\mid r+ik\le x\le r+(i+1)k, s+jk\le y \le s+(j+1)k\}$ for
some integers $i$ and $j$. Let the optimal solution of $(r,s)$-shifting be the union sets of all the optimal
solution of the subdivided grids, and $w(OPT(r,s))$ be the weight of the optimal solution of $(r,s)$-shifting.
It can be
proven that there is at least one $(r, s)$-shifting, $0 \le r, s \le
k-1$, with
\begin{equation}
w(OPT(r,s))\ge (1-\frac{1}{k})^2w(OPT(\mathcal{B}))
\end{equation}
where $OPT(\mathcal{B})$ is the maximum weighted independent set of all the bidders, and $w(OPT(\mathcal{B}))$ is the weight of
$OPT(\mathcal{B})$. For any given integer $k\ge 1$, there are $k^2$ kinds of different shiftings in total. We will choose the
optimal solution of $(r,s)$-shifting's that with the highest weight
as our final approximation solution. Thus, we have a PTAS
for optimal channel allocation problem,  \ie setting
$k=\frac{1+\epsilon+\sqrt{1+\epsilon}}{\epsilon}$.

Based on this PTAS we then present our channel allocation mechanism
 with privacy preserving.
Observe that the bidders  submit their bids to agent encrypted
 using the auctioneer's public key.
Following the PTAS protocol, we need to compute a maximum
 weighted independent set for each grid in the $(r,s)$-shifting, \ie,
 compare the weights of all independent sets.
Clearly, the auctioneer should not access the encrypted bid of any
 bidder as he has the decryption key.
In our protocol, the agent will compute $E(\sum_{i \in S} b_i)$ for
 each of the maximal independent set contained in a grid,
 which can be done easily as $E(b_i)$ is computed from homomorphic encryption.
For any given grid $g_j^{r,s}$ of the $(r,s)$-shifting, let $\mathcal{D}=\{d_{1,j}^{r,s},\cdots,d_{z,j}^{r,s}\}$ be the
set of \emph{maximal independent sets} of bidders in $g_j^{r,s}$. We use $OPT(g_j^{r,s})$ to denote the optimal solution
in the grid $g_j^{r,s}$.
Clearly $\mathcal{D}$ has cardinality of at most $O(k^2)$ and can be
enumerated in time $O(n^{O(k^2)})$.
In Algorithm~\ref{alg:2}, we present our method for finding the
 $OPT(g_j^{r,s})$ for each subdivided grid $g_j^{r,s}$ with
 privacy preserving.
To hide the true values of $w(d_{i,j}^{r,s})$ (which may break
 privacy) from the auctioneer, the agent will mask them by using two
 random values $\delta_1$ and $\delta_2$ as $\delta_1+ \delta_2 \cdot
 w(d_{i,j}^{r,s})$.
Note that the range $[1,2^{\gamma_1}]$ and $[1,2^{\gamma_2}]$ for $\delta_1$ and
 $\delta_2$ are chosen based on
 the consideration of the correctness of modular operations:
 $\delta_1 + \delta_2 \cdot w(d_{i,j}^{r,s})$ should be smaller than
 the modulo used in Paillier's system.

Assume that the number of grids that subdivided by $(r,s)$-shifting is $N_{r,s}$, then the optimal solution of
$(r,s)$-shifting is $OPT(r,s)=\bigcup \nolimits_{j\le
  N_{r,s}}{d_{\sigma(1),j}^{r,s}}$. By sending the intermediate results
to the auctioneer, the auctioneer can compare and find which
independent set will be chosen for each subgrid.
Observe that both the auctioneer and the agent will not know the bid
 values in the independent set.
By using the optimal
solution of each grid, the agent can calculate the encrypted value $E(w(OPT(r,s)))$, and allocate channels to bidders without
leaking the true bid values of bidders.
The allocation will be sent to the auctioneer.
The details are described in Algorithm \ref{alg:3}.

\begin{algorithm}
\caption{Computing the optimal solution for grid $g_j^{r,s}$}\label{alg:2}
{
\begin{algorithmic}[1]

\STATE The agent randomly picks two integers $\delta_1 \in \mathbb{Z}_{2^{\gamma_{1}}}$, $\delta_2 \in \mathbb{Z}_{2^{\gamma_{2}}}$,
computes and sends $\{E(\delta_1+\delta_2w(d_{i,j}^{r,s}))\}_{1\le i\le z}$ to the auctioneer, where
\begin{equation*}
E(\delta_1+\delta_2w(d_{i,j}^{r,s}))=E(\delta_1)(\prod\nolimits_{l\in d_{i,j}^{r,s}}{E(b_l)})^{\delta_2}
\end{equation*}

\STATE The auctioneer decrypts
$\{E(\delta_1+\delta_2w(d_{i,j}^{r,s}))\}_{0\le i\le z}$, and sorts
them in non-increasing order. Assume
\begin{equation*}
w(d_{\sigma(1),j}^{r,s})\ge w(d_{\sigma(2),j}^{r,s}) \ge... \ge w(d_{\sigma(z),j}^{r,s})
\end{equation*}
where $d_{\sigma(i),j}^{r,s}$ is the maximum independent set with rank $i$ in the sorted  list.

\STATE The auctioneer sends $\{\sigma(i)\}_{1\le i \le z}$ to the agent.
\STATE The agent chooses $d_{\sigma(1),j}^{r,s}$ as the optimal solution of grid $g_j^{r,s}$.
\end{algorithmic}
}
\end{algorithm}

\begin{algorithm}
\caption{PTAS with bid privacy preserving}\label{alg:3}
{
\begin{algorithmic}[1]

\STATE The agent randomly picks two integers $\delta_3 \in \mathbb{Z}_{2^{\gamma_{1}}}$, $\delta_4 \in \mathbb{Z}_{2^{\gamma_{2}}}$,
computes and sends $E(\delta_3+ \delta_4 w(OPT(r,s)))$ for any $1\le r,s \le k$ to the auctioneer, where
\begin{equation*}
E(\delta_3+\delta_4 w(OPT(r,s)))=E(\delta_3)(\prod\limits_{j\le N_{r,s}}{E(w(d_{\sigma(1),j}^{r,s}))})^{\delta_4}
\end{equation*}

\STATE The auctioneer decrypts and sorts the weights of the optimal solution of different shiftings in non-increasing order.
\begin{equation*}
w(OPT(\sigma_1(1),\sigma_2(1)))\ge ... \ge w(OPT(\sigma_1(k^2),\sigma_2(k^2)))
\end{equation*}
where $OPT(\sigma_1(i),\sigma_2(i))$ is the optimal solution of ($\sigma_1(i),\sigma_2(i)$)-shifting with rank $i$ in the sorted  list.

\STATE The auctioneer sends $\{(\sigma_1(i),\sigma_2(i))\}_{1\le i \le k^2}$ to the agent.
\STATE The agent chooses $OPT(\sigma_1(1),\sigma_2(1)))$ as the final solution, and sends the allocation result to the auctioneer.
\end{algorithmic}
}
\end{algorithm}

\begin{lemma}\label{monotone}
 Our  allocation mechanism for SUA is bid-monotone.
\end{lemma}

\begin{IEEEproof}
Without loss of generality, we assume that the bidder $i$ wins by bidding $b_i$ in grid $g_{j}^{r,s}$.
Then, $\sigma_1(1)=r$, $\sigma_2(1)=s$ and bidder $i$ in $d_{\sigma(1),j}^{r,s}$. It is not hard to get that the bidder
$i$ is still in $d_{\sigma(1),j}^{r,s}$ when he increases his bid to $b'_i> b_i$. Furthermore, the increased weight
of other shiftings is no more than $(r,s)$-shifting when $i$ increases his bid, which indicates that $\sigma_1(1)=r$ and $\sigma_2(1)=s$ still hold.
Thus, we can conclude that $i$ will always win by bidding $b'_{i}> b_i$.
\end{IEEEproof}

\subsection{\textbf{Payment Calculation with Privacy Preserving}}

We have proved that our allocation mechanism is bid-monotone, which indicates that there exists a
critical value for each bidder.
The bidder $i$ will win the auction by bidding a price
which is higher than its critical value,
otherwise, bidder $i$ will lose in the auction.
To ensure the strategyproofness of our auction mechanism, we will compute the critical value for each winner
as the final payment in the following.

Without loss of generality, we also assume that the bidder $i$ wins by bidding $b_i$ in grid $g_{j}^{r,s}$.
We further assume
that $d_{l(i),j}^{r,s}$ is the maximum independent set with highest weight which does not include bidder $i$, and
$OPT(\sigma_1(f(i)),\sigma_2(f(i)))$ is the optimal solution of $(\sigma_1(f(i)),\sigma_2(f(i)))$-shifting
which has the highest weight and does not
include the bidder $i$.
We will calculate the critical value of the winner $i$ based on the following considerations.
\begin{itemize}
  \item The minimum bid price, denoted as $p_i^1$,   ensures bidder $i$   win in grid $g_{j}^{r,s}$. Then, we can get that
  \begin{equation*}
  p_i^1=w(d_{l(i),j}^{r,s})-w(d_{\sigma(1),j}^{r,s})+b_i
  \end{equation*}

  \item The minimum bid of bidder $i$ which makes $OPT(r,s)$ always with the highest weight among all the optimal
  solutions of shiftings including bidder $i$. We use $p_i^2$ ($p_i^2$ exists \emph{iff} $f(i)>2$) to denote this minimum bid, and set $p_i^{2,q}=w(OPT(\sigma_1(q),\sigma_2(q)))-w(d_{\sigma(1),j}^{\sigma_1(q),\sigma_2(q)})
  +w(d_{l(i),j}^{\sigma_1(q),\sigma_2(q)})-w(OPT(r,s))+b_i$, then
  \begin{eqnarray*}
  p_i^2=\max \{p_i^{2,1},...,p_i^{2,f(i)-2}\}
  \end{eqnarray*}

  \item The minimum bid of bidder $i$ that ensures $w(OPT(r,s))\ge w(OPT(\sigma_1(f(i)),\sigma_2(f(i))))$, which is denoted by
  $p_i^3$. Then, we can get that
  \begin{equation*}
  p_i^3=w(OPT(\sigma_1(f(i)),\sigma_2(f(i))))-w(OPT(r,s))+b_i
  \end{equation*}

\end{itemize}

In conclusion, the critical value of bidder $i$ is $p_i=\max\{p_i^1,p_i^2,p_i^3,0\}$.
 Since the agent knows the order of all the maximum independent sets of each grid and the order of all the
 optimal solution of shiftings, he can compute the encrypted value of $p_i^1$, $p_i^2$ and $p_i^3$ by
 homomorphic operations, respectively. Then, our payment calculation mechanism with privacy preserving is depicted as follows:
 \begin{enumerate}
   \item The agent computes $E(p_i^1)$, $E(p_i^{2,1})$, $\cdots$,
     $E(p_i^{2,f(i)-2})$, $E(p_i^3)$, and sends the results to the auctioneer.
   \item The auctioneer decrypts the ciphertexts and sets the payment
     of winner $i$ as
   \begin{equation*}
   p_i=\max\{p_i^1,p_i^{2,1},...,p_i^{2,f(i)-2},p_i^3,0\}
   \end{equation*}
 \end{enumerate}

It is easy to prove the following   theorems.
\begin{theorem}\label{critical}
PPS-SUA charges each winner its critical value and is strategyproof.
\end{theorem}

\begin{theorem}
The computation and communication cost of PPS-SUA are all $O(n^{k^{2}+1})$.
\end{theorem}

\subsection{\textbf{Privacy analysis of PPS-SUA}}

\begin{theorem}\label{bidprivacyTheorem}
PPS-SUA is  bid privacy-preserving.
\end{theorem}
\begin{proof}
To confirm the bid privacy, we consider the view of
agent and auctioneer, respectively.

During our auction mechanism for SUA, the agent can obtain nothing
but the encrypted bids and the sorting results of the weight of each grid
and each shifting. Based on the IND-CPA security of homomorphic
cryptosystem, the agent cannot learn more information about the bid of
any bidder.

The auctioneer holds the decryption key. Nevertheless, he has no direct
access to the encrypted bids. While computing the optimal allocation
and critical value of winner $i$, the auctioneer can receive the encrypted
 weight of maximal independent sets in each grid, weight of the optimal solution of each
 shifting, and $\{p_i^1,p_i^{2,1},...,p_i^{2,f(i)-2},p_i^3\}$.
 From the weight of solutions in the grids or shiftings, the auctioneer
cannot infer any bid, since they are encrypted by the agent and the auctioneer has
no idea about which bidders are in these solutions, except the winning shifting.
Consider $\{p_i^1,p_i^{2,1},...,p_i^{2,f(i)-2},p_i^3\}$, auctioneer can construct
the equation of them. However, the bid value of bidder $i$ can still be well preserved, as auctioneer does
not know any value of the variables in these equations.
\end{proof}

\section{A Multi-unit Scheme}\label{sec:heterogeneous}
In this section, we propose a strategyproof auction mechanism
for MUA, namely PPS-MUA, which maximizes the social efficiency and protects
the bid privacy of bidders.
Then, we  design an extended version of PPS-MUA, namely PPS-EMUA, to improve the
average performance of PPS-MUA.

\subsection{\textbf{Initialization and Bidding}}

The initialization and bidding procedure in MUA is similar as that in SUA,
which can be referred in section \ref{sec:suaini}.
At last, each bidder $i$ encrypts his bid $b_i$ by using the encryption key of the auctioneer,
and only sends $(E(b_i), N_i, L_i)$ to the agent.

\subsection{\textbf{Allocation Mechanism with Privacy Preserving}}

Since SUA is a special case of MUA, the optimal allocation issue in MUA is also
NP-hard.
Thus, we will introduce a simple allocation mechanism which
approximates the social efficiency.
We first subdivide the plane into grids at size $2*2$,
and use the symbol $g^l$ to denote the $l$-th $2*2$ grid.
It is obvious that there are four
$1*1$ grids in each $g^l$.
These four $1*1$ grids can be categorized into four types as shown in Fig. $2(a)$. Let $g_r^l$ be the $1*1$
grid in $g^l$ with type $r$, $g_r$ be the set of $1*1$ grids with type $r$.
We also assume that
the conflict radius of each bidder is $\frac{1}{2}$ and
regard each bidder as a unit disk.
Obviously, each bidder  located in $g_r^l$ cannot conflict with the
bidders  located in $g_r^{l'}$ when $l\neq l'$.
Let $OPT(g_r^l)$ be the optimal solution of allocation problem in
$g_r^l$, $OPT(g_r)$ be the optimal solution of the
allocation problem in $g_r$, then $OPT(g_r)=\bigcup \nolimits_{l}{OPT(g_r^l)}$.

\input{grid.TpX}

Note that we cannot get the optimal solution in each grid
$g_r^l$. To tackle this, we further subdivide each $1*1$
grid $g_r^l$ into four $\frac{1}{2}*\frac{1}{2}$ sub-grids as shown in Fig. $2(b)$, which are denoted by $g_{r,1}^l$,
$g_{r,2}^l$, $g_{r,3}^l$ and $g_{r,4}^l$, separately. Notice that all the bidders located in the same
sub-grid $g_{r,s}^l$ conflict with each other. Thus, one channel can only be sold to one bidder in $g_{r,s}^l$.

The optimal allocation problem in each sub-grid $g_{r,s}^l$ can be reduced to a \emph{knapsack problem} (KP).
Although the KP is an NP-hard
problem, there exists a PTAS \cite{lai2006knapsack}, and a greedy allocation mechanism with
approximation factor of $2$ (the details can be referred to lines $3$-$5$ in Algorithm \ref{alg:5}).
It is hard to design a privacy preserving version of the PTAS based on dynamic programming, thus,
we design our allocation mechanism for MUA based on the greedy allocation mechanism in
each sub-grid $g_{r,s}^l$. Assume that $APP(\mathcal{B})$, $APP(g_r)$, $APP(g_r^l)$ and
$APP(g_{r,s}^l)$ are the approximation solution of the
allocation problem in the whole plane, $g_r, g_r^l$ and $g_{r,s}^l$, separately.
 We choose the $APP(g_{r,s}^l)$ with biggest weight as the solution of grid $g_{r}^l$ and the
$APP(g_r)$ with the biggest weight as our final solution $APP(\mathcal{B})$ (the details is depicted in Algorithm \ref{alg:5}).

\begin{theorem}
 Our auction mechanism for MUA has an approximation factor of $32$.
\end{theorem}
\begin{proof}
Assume that $OPT(\mathcal{B})$ is the optimal solution of our original allocation problem,
and $OPT_r(\mathcal{B})=\{i|i\in OPT(\mathcal{B})$ and $i$ is allocated in $g_r\}$. Then, we can get that
\begin{eqnarray*}
&w(OPT(\mathcal{B}))&=\sum \limits_{1\le r \le 4} {w(OPT_r(\mathcal{B}))} \le \sum \limits_{1\le r \le 4} {w(OPT(g_r))} \\
&& \le 4\max \{w(OPT(g_r))\}_{1\le r \le 4}
\end{eqnarray*}
where $w(\cdot)$ is an operation to compute the weight of solutions.
For each grid $g_r^l$, we can get that
$w(OPT(g_r^l)) \le \sum \limits_{1\le s \le 4} {w(OPT(g_{r,s}^l))}
 \le 4\max \{w(OPT(g_{r,s}^l))\}_{1\le s \le 4}$.

Since we sort bidders in non-increasing order according to their
per-unit bidding
prices, so user $i$ has the $i$-$th$ largest value in $\frac{b_i}{N_i}$
and $\sum\nolimits_{i=0}^{k}{N_i}>m$,
$\sum \nolimits_{i=0}^{k}{b_i}>w(OPT(g_{r,s}^l))$. Our approximation allocation mechanism sets $APP(g_{r,s}^l)=\{1, 2, ..., k-1\}$ if
$\sum \nolimits_{i=0}^{k-1}{b_{i}}\ge b_k$; otherwise, we set $APP(g_{r,s}^l)=\{k\}$. Thus, $OPT(g_{r,s}^l)\le 2 APP(g_{r,s}^l)$.
Because we choose the $APP(g_{r,s}^l)$ with biggest weight as $APP(g_{r}^l$), we can further get that
$OPT(g_{r}^l)\le 4 \max(OPT(g_{r,s}^l))_{1\le s \le 4}\le 8 APP(g_{r}^l)$. In a similar way, we can get that
$OPT(\mathcal{B})\le 4\max(OPT(g_{r}))_{1\le r \le 4}\le 32 APP(\mathcal{B})$.
\end{proof}

\begin{algorithm}
\caption{Channel allocation mechanism for MUA}\label{alg:5}
{
\begin{algorithmic}[1]
\FOR {each sub-grid $g_{r,s}^l$}
\IF {The number of channels that all the bidders located in $g_{r,s}^l$ want to buy is larger than $m$ }
\STATE Sorting the bidders that located in $g_{r,s}^l$ in non-increasing order according to their per-unit bid values $\frac{b_i}{N_i}$,
where $\sigma(i)$ is the bidder with $i$-th per-unit bid value in the sorted list;
\STATE Find the critical bidder $\sigma(k)$ in the sorted bidder list, which satisfies:
\begin{equation*}
\sum \nolimits_{i=1}^{k-1}{N_{\sigma(i)}} \le m < \sum \nolimits_{i=1}^k{N_{\sigma(i)}};
\end{equation*}
\STATE Set $APP(g_{r,s}^l)=\{\sigma(1), \sigma(2), ..., \sigma(k-1)\}$ if $\sum \nolimits_{i=1}^{k-1}{b_{\sigma(i)}}\ge b_{\sigma(k)}$; otherwise,
set $APP(g_{r,s}^l)=\{\sigma(k)\}$;
\ELSE
\STATE Set $APP(g_{r,s}^l)$ is all the bidders that located in $g_{r,s}^l$;
\ENDIF
\ENDFOR
\FOR {each grid $g_r^l$}
\STATE Set $s'= \arg \max \limits_{s}\{w(APP(g_{r,s}^l))|1\le s\le 4\} $, where $w(\cdot)$ is an operation to compute the weight of solutions.
\STATE Set $APP(g_r^l)=APP(g_{r,s'}^l)$;
\ENDFOR
\FOR {$r=1$ to $4$}
\STATE Set $APP(g_r)=\bigcup_{l}{APP(g_{r}^l)}$;
\ENDFOR
\STATE Set $r'=\arg \max \limits_{r}\{w(APP(g_r))|1\le r\le 4\} $;
\STATE Return $APP(\mathcal(B))=APP(g_{r'})$ as the final solution;
\end{algorithmic}
}
\end{algorithm}

In order to protect the true bid value of bidders, the agent confuses the ID of bidders by using a permutation
$\pi: \mathbb{Z}_n \rightarrow \mathbb{Z}_n$ after receiving the encrypted bid of bidders. Then, the privacy
preserving version of our approximation allocation mechanism is
depicted  in Algorithm \ref{alg:6}.
\begin{lemma}\label{lemma:muamonotone}
 Our allocation mechanism for MUA is bid-monotone.
\end{lemma}
\begin{proof}
Assume bidder $i$ is located in grid $g_{r,s}^l$ and wins the auction by bidding $b_i$, then he must be in the solutions
$APP(g_{r,s}^l), APP(g_r^l)$ and $APP(\mathcal{B})$ at the same time.
Thus, we will check if the bidder $i$ still belongs to these
solutions when he bids $b'_i>b_i$ in the following.

First, we consider the solution $APP(g_{r,s}^l)$. Obviously, the rank of bidder $i$ will not decrease
 when bidder $i$ increases his bidding value. Thus, $b'_i$ is always larger than the sum bid of the top $k-1$ bidders
 when $i=\sigma(k)$, which means $i$ will remain in $APP(g_{r,s}^l)$ in this case.
 In another case, all the bidders with top
 ($k-1$) per-unit bid remains unchanged when $i$ bids $b'_i>b_i$, and thus their sum bid is still larger than the $k$-th bid.
 Thus, $i$ will always win the auction when he increases his bid.

Then, we consider the solutions $APP(g_{r}^l)$ and $APP(\mathcal{B})$.
When  $i$ bids $b'_i>b_i$, the $w(APP(g_{r,s}^l))$ will
increase, and $w(APP(g_{r,s'}^l))$ will keep unchanged if $s'\neq s$.
Thus, $APP(g_{r,s}^l)$ still has the highest weight and will be
selected as $APP(g_{r}^l)$.
Similarly,  $APP(g_{r})$ will be selected as the final allocation
$APP(\mathcal{B})$ either.

Bidder $i$ will always win by bidding $b'_i>b_i$ if he wins by bidding
$b_i$,  \ie, our allocation mechanism is bid-monotone.
\end{proof}

\begin{algorithm}
\caption{Channel allocation mechanism for MUA with bid privacy}\label{alg:6}
{
\begin{algorithmic}[1]
\FOR {each sub-grid $g_{r,s}^l$}
\IF {The number of channels that all the bidders located in $g_{r,s}^l$ want to buy is larger than $m$ }
\STATE The agent randomly chooses two integers $\delta_{r,1}^l \in \mathbb{Z}_{2^{\gamma_{1}}}$, $\delta_{r,2}^l \in \mathbb{Z}_{2^{\gamma_{2}}}$,
computes and sends $(\pi(i), E(\delta_{r,1}^l b_i+\delta_{r,2}^l),N_i)$ to the auctioneer if $i$ is located in $g_{r,s}^l$.
\STATE The auctioneer decrypts and sorts the per-unit bids of bidders in non-increasing order;
\STATE The auctioneer finds the critical bidder $\sigma(k)$ in the sorted bidder list, and sends $(\{\sigma(i)\}_{i<k}, \sigma(k))$
to the agent;
\STATE The agent computes and sends $E(\delta_{r,1}^l \sum \nolimits_{i=1}^{k-1}{b_{\sigma(i)}}+\delta_{r,2}^l))$
to the auctioneer;
\STATE The auctioneer sends $\{\sigma(i)\}_{i<k}$ to the agent if $\sum \nolimits_{i=1}^{k-1}{b_{\sigma(i)}}\ge b_{\sigma(k)}$; otherwise, he sends $\sigma(k)$;
\STATE The agent sets $APP(g_{r,s}^l)$ includes all the bidders that the auctioneer sent to him;
\ELSE
\STATE The agent sets $APP(g_{r,s}^l)$ as all the bidders  located in $g_{r,s}^l$;
\ENDIF
\ENDFOR
\FOR {each grid $g_r^l$}
\STATE The agent chooses two integers $\delta_{r,3}^l \in \mathbb{Z}_{2^{\gamma_{1}}}$, $\delta_{r,4}^l \in \mathbb{Z}_{2^{\gamma_{2}}}$,
computes $\{(s, E(\delta_{r,3}^l w(APP(g_{r,s}^l))+\delta_{r,4}^l)\}_{1\le s \le 4}$ and sends them to the auctioneer.
\STATE The auctioneer decrypts the ciphertexts and finds $s'= \arg \max \limits_{s}\{w(APP(g_{r,s}^l))|1\le s\le 4\} $. Then,
he sends $s'$ to the agent.
\STATE The agent sets $APP(g_r^l)=APP(g_{r,s'}^l)$;
\ENDFOR
\FOR {$r=1$ to $4$}
\STATE The agent sets $APP(g_r)=\bigcup_{l}{APP(g_{r}^l)}$;
\ENDFOR
\STATE The agent chooses two integers $\delta_{1}^l \in \mathbb{Z}_{2^{\gamma_{1}}}$, $\delta_{2}^l \in \mathbb{Z}_{2^{\gamma_{2}}}$,
computes $\{(r, E(\delta_{1}^l w(APP(g_r))+\delta_{2}^l)\}_{1\le r \le 4}$ and sends them to the auctioneer.
\STATE The auctioneer decrypts the ciphertexts and finds $r'=\arg \max \limits_{r}\{w(APP(g_r))|1\le r\le 4\} $. Then,
he sends $r'$ to the agent;
\STATE The agent sets $APP(\mathcal{B})=APP(g_{r'})$, and sends $APP(\mathcal{B})$ to the auctioneer as the final solution;
\end{algorithmic}
}
\end{algorithm}

\subsection{\textbf{Payment Calculation with Privacy Preserving}}

We now consider the procedure of payment calculation for a winner
$i$ which is located in grid $g_{r,s}^l$.

Since the bidder $i$ wins the auction, we can conclude that:
1) $i\in APP(g_{r,s}^l)$;
2) $APP(g_{r}^l)=APP(g_{r,s}^l)$; and 3) $APP(\mathcal{B})=APP(g_r)$. We first consider the minimum bid value
of bidder $i$, denoted by $p_i^1$, with which the bidder $i$ will be put in $APP(g_{r,s}^l)$.
In the case that all the bidders
located in $g_{r,s}^l$ win the auction, we set $p_i^1=0$; otherwise, we assume that $i=\sigma(j)$
in the sorted bidder list of $g_{r,s}^l$ when $i$ bids $b_i$, then the process of $p_i^1$ computation
is shown in Algorithm \ref{alg:7}.

Under the assumption that $APP(g_{r,s}^l)$ keeps unchanged, we suppose $p_i^2$ is the minimum bid value of bidder $i$
that makes $APP(g_{r}^l)=APP(g_{r,s}^l)$, $p_i^3$ is the minimum bid value of bidder $i$ that makes
$APP(\mathcal{B})=APP(g_r)$. Then, we have
\begin{equation*}
p_i^2=\max\{w(APP(g_{r,s'}^l))|s'\neq s\}-w(APP(g_{r,s}^l))+b_i
\end{equation*}
\begin{equation*}
p_i^3=\max\{w(APP(g_{r'}))|r'\neq r\}-w(APP(g_r))+b_i
\end{equation*}

The critical value of bidder $i$ is $p_i=\max(p_i^1, p_i^2, p_i^3)$. Next we will show that
we can compute the critical value for each winner without leaking the true bid value of bidders.

\begin{algorithm}
\caption{$p_i^1$ computation for winner $i$ in MUA}\label{alg:7}
{
\begin{algorithmic}[1]
\STATE Set $j=j+1$;
\STATE Set $b'_i=\frac{b_{\sigma(j)}N_i}{N_{\sigma(j)}}$;
\STATE Run lines $3\thicksim 5$ of Algorithm \ref{alg:5} to check if bidder $i$ will win by bidding $b'_i$;
\IF {$i$ wins by bidding $b'_i$}
\STATE Repeat steps $1 \thicksim 3$ until $i$ lose the auction;
\ENDIF
\IF {$i$ is the $k$-th bidder when he bids $b'_i$}
\STATE Set $p_i^1=\max (\sum \nolimits_{q=1}^{k-1}{b_{\sigma'(q)}}, b'_i)$, where $\sigma'(q)$ is the bidder
with $q$-th per-unit bid when $i$ bids $b'_i$;
\ELSE
\STATE Set $p_i^1=\max (b_{\sigma'(k)}+b_i-\sum \nolimits_{q=1}^{k-1}{b_{\sigma'(q)}}, b'_i)$;
\ENDIF
\end{algorithmic}
}
\end{algorithm}

Since the agent can compute $E(\delta_{r,1}^l b'_i N_{\sigma(j)}+\delta_{r,2}^l N_{\sigma(j)})$ which is
equal to $E(\delta_{r,1}^l b_{\sigma(j)} N_{i}+\delta_{r,2}^l N_{\sigma(j)})$, the auctioneer can decrypt
and compute the value of $\delta_{r,1}^l b'_i+\delta_{r,2}^l$. Thus, the auctioneer and agent can check
if bidder $i$ will win the auction by bidding $b'_i$ as they did in lines $3\thicksim 7$ of Algorithm \ref{alg:6}.
Further, the agent can get $\max\{w(APP(g_{r,s'}^l))|s'\neq s\}$ and $\max\{w(APP(g_{r'}))|r'\neq r\}$ via
 communicating with the auctioneer. Thus, the agent can choose two integers $\delta_1 \in \mathbb{Z}_{2^{\gamma_{1}}}$,
  $\delta_2 \in \mathbb{Z}_{2^{\gamma_{2}}}$ and compute the ciphertexts of $\delta_1 p_i^1+\delta_2$,
 $\delta_1 p_i^2+\delta_2$ and $\delta_1 p_i^3+\delta_2$
 through homomorphic operations, and sends them to the auctioneer. Then, the auctioneer decrypts these
ciphertexts, sets $\delta_1 p_i+\delta_2=\max (\delta_1 p_i^1+\delta_2,\delta_1 p_i^2+\delta_2,\delta_1 p_i^3+\delta_2)$
 and sends $\delta_1 p_i+\delta_2$ to the agent. After computing the payment $p_i$ of each winner $i$,
 the agent sends them to the auctioneer.

From above analysis, we can conclude that:
\begin{theorem}\label{lemma:muacritical}
We charge each winner its critical value in PPS-MUA.
PPS-MUA is strategyproofness.
\end{theorem}

\subsection{\textbf{Extended Auction Mechanism for MUA}}

We have designed a simple allocation mechanism for MUA, which provides an approximation factor of $32$.
However, PPS-MUA only chooses the solution
of a $\frac{1}{2}*\frac{1}{2}$ sub-grid as the final solution of a $2*2$ grid,
while dropping all the other bidders that located in other $15$ sub-grids.
Although the allocation in this way provides a guarantee
for the worst case performance, the average performance may be relatively low.
To address this issue, we    extend
 our allocation mechanism by supplementing the solution with other
 bidders as shown in Algorithm \ref{alg:8}.

\begin{algorithm}
\caption{Extended Allocation Mechanism PPS-EMUA}\label{alg:8}
{
\begin{algorithmic}[1]
\STATE Run Algorithm \ref{alg:5} to allocate channels to bidders;
\STATE Sort all the bidders who lose in Algorithm \ref{alg:5} in non-increasing order according to their bid values.
\FOR {each loser $i$ in the sorted list}
\IF {we can allocate channels to $i$ without interfering with the existing winners}
\STATE Set $i$ wins and allocate channels to him;
\ENDIF
\ENDFOR
\end{algorithmic}
}
\end{algorithm}

\begin{lemma}
 The allocation mechanism PPS-EMUA presented in Algorithm~\ref{alg:8} is bid-monotone.
\end{lemma}
\begin{proof}
Since we have proved that if the winner $i$ increases his bid in Algorithm \ref{alg:5}, he will
always win the auction. Here, we only need to concentrate on the winners that lose in Algorithm \ref{alg:5}, but will
win in the extended version. Suppose such a winner $i$ increases
his bid to $b'_i$ which satisfies $b'_i>b_i$, there are two possible cases: 1) $i$ wins in Algorithm \ref{alg:5} and 2) $i$ remains lose in Algorithm \ref{alg:5}.
In the case that $i$ loses in Algorithm \ref{alg:5}, the final allocation of Algorithm \ref{alg:5} is the
same as the allocation when $i$ bids $b_i$. Thus, there is no new bidder whose bidding price is higher than $i$ in the
sorted loser list of Algorithm \ref{alg:8} after the bidder $i$ increasing his bid.
In addition to $i$ wins by bidding $b_i$, we
can conclude that the bidder $i$ will also win the auction when he increases his bid.
\end{proof}

As this new allocation mechanism is bid-monotone, there exists a critical value for
each winner. We use $p'_i$ here to denote the minimum bid value of bidder $i$ with which $i$ will win
in Algorithm \ref{alg:5}, and $p''_i$ to denote the minimum bid value of winner $i$ with which $i$ will
 win in the sorted loser list. According to Algorithm \ref{alg:8}, $p'_i$ is the critical value of
 bidder $i$ in Algorithm \ref{alg:5}, and $p''_i$ should be smaller than $p'_i$.

For each winner $i$, his critical value can be computed as follows:
\begin{itemize}
  \item If $i$ wins in line $1$ of Algorithm \ref{alg:8} and will
  lose as long as he bids $b'_i<p'_i$, his critical value is equal to $p'_i$;
  \item Otherwise, his critical value is equal to $p''_i$. Suppose $f(i)$
  is the first bidder in the sorted loser list who loses the auction but
  will win as long as the bidder $i$'s bidding price is smaller than his,
  then $p''_i=b_{f(i)}$ if $f(i)$ exits and $p''_i=0$ otherwise.
\end{itemize}

As the extended allocation mechanism is bid-monotone and we always
charge each winner its critical value, we have
\begin{theorem}
PPS-EMUA is strategyproof and social efficient.
\end{theorem}

In the following, we will show that PPS-EMUA can  be performed with privacy preserving.
Due to the page limit, we will only briefly introduce our ideas. Algorithm \ref{alg:9}
shows the allocation mechanism of PPS-EMUA with bid privacy.

\begin{algorithm}
\caption{PPS-EMUA: Privacy-Preserving Allocation Mechanism}\label{alg:9}
{
\begin{algorithmic}[1]
\STATE The auctioneer and the agent run Algorithm \ref{alg:6};
\STATE The agent randomly chooses two integers $\delta_1 \in \mathbb{Z}_{2^{\gamma_{1}}}$,
$\delta_2 \in \mathbb{Z}_{2^{\gamma_{2}}}$, computes and sends
$(\pi(i), E(\delta_1 b_i+\delta_2))$ if bidder $i$ loses in Algorithm \ref{alg:6},
and $\{\pi(i), N_i, L_i\}_{i\in \mathcal{B}}$ to the auctioneer;
\STATE The auctioneer decrypts the encrypted bids, and run lines
$2 \thicksim 5$ of Algorithm \ref{alg:8};
\end{algorithmic}
}
\end{algorithm}

The procedure of payment calculation has four steps: 1) We can obtain $p'_i$ for
each winner who wins in line $1$ of Algorithm \ref{alg:8}
and protect the true bid value of bidders by using the method we have introduced previously.
2) The auctioneer and agent can
check if bidder $i$ will lose as long as his bid is smaller than $p'_i$ by running Algorithm \ref{alg:9}
 and assuming $i$ loses in line $1$ of Algorithm \ref{alg:8}. 3) In the case that $i$ may win when he bids smaller than
$p'_i$, the auctioneer sets $p''_i=0$ if $f(i)$ does
not exist, and sets $p''_i=\delta_1 b_{f(i)}+\delta_2$ if $f(i)$ exists.
The auctioneer sends $\delta_1 p'_i+\delta_2$ in the case that $p'_i$
 is the critical value of bidder $i$, and $\delta_1 p''_i+\delta_2$ in other case.
 4) With the encrypted critical value, the agent can compute the payment of winner $i$.
 After obtaining all the payment of winners, the agent will send them to the auctioneer.

\begin{theorem}
The computation and communication cost are all $O(n^2)$ for PPS-MUA and PPS-EMUA.
\end{theorem}

\subsection{\textbf{Privacy Analysis}}

\begin{theorem}\label{bidprivacyTRAP-EMUA}
PPS-MUA and PPS-EMUA are   privacy-preserving for each bidder.
\end{theorem}
\begin{proof}
Here we only prove it for PPS-EMUA as PPS-MUA is a procedure of PPS-EMUA.
We first consider the agent. Except the encrypted bids, the agent can only obtain some orders,
such as the bidding price of the bidders in each sub-grid,
during our auction mechanism of PPS-EMUA. In the process of payment calculation,
the agent can get nothing but the auction outcomes and some new orders.
Based on the IND-CPA security of homomorphic cryptosystem, the agent cannot learn more information about the bid of
any bidder.

Although the auctioneer holds the decryption key, he has no direct
access to the encrypted bids. While computing the allocation in each sub-grid $g_{r,s}^l$,
the auctioneer can built $|g_{r,s}^l|+1$
functions that with $|g_{r,s}^l|$ bids and two random numbers, where $|g_{r,s}^l|$ is the
number of bidders that located in $g_{r,s}^l$. Since the
number of variables is larger than the number of functions, the auctioneer cannot decrypt
any true bid value of bidders. In the other
parts of our auction mechanism, the auctioneer only receives the weight of solutions.
Since the auctioneer has no idea
about which bidders are in these solutions, he can also get nothing from them.
\end{proof}


\section{Performance Evaluations}\label{sec:simulation}

\subsection{Simulation Setup}

In our simulations, the number of bidders varies from 50 to 300, and
all the bidders are randomly distributed in a square area.
The bidding price of each bidder is uniformly generated in $[0,100]$.
We use a $1024$-bit length Paillier¡¯s homomorphic encryption
system in the simulation.
Thus, we choose $\gamma_1 = 1007$ and $\gamma_2 =1022$ to ensure the correctness of modular operations.
For Multi-Unit Auction (MUA), we  assume the channel demand of each bidder
is randomly generated from 1 to 4, and there are 4 or 8 available channels in spectrum market.

We mainly study the social efficiency ratio, computation overhead and the communication overhead in our simulations. We define the
social efficiency ratio is the ratio between the social efficiency of our approximation mechanism and the optimal one.
Since agent and auctioneer are two central party in this paper, we evaluate the computation overhead of them
in our design by recording the required processing time, and evaluate the communication overhead through calculating
the size of essential information transferred in the auction. All the simulations are performed over 100 runs and the result is the
averaged value.

\begin{figure*}[!tbp]
    \centering
    \subfigure[The Performance of Social Efficiency Ratio under SUA model]{
        \label{fig:model1social}
        \includegraphics[width=55mm,height=40mm]{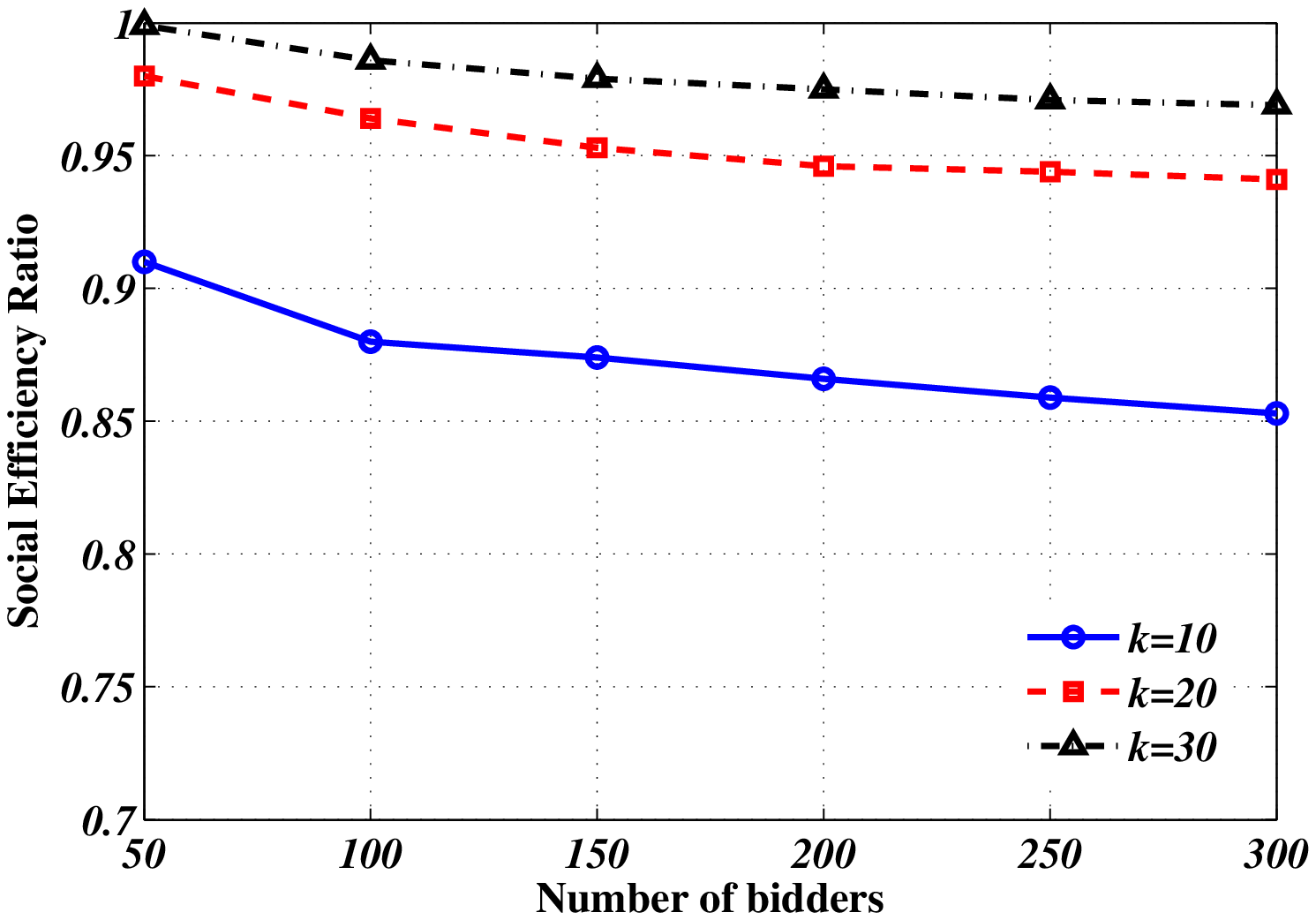}
    }
    \hspace{1mm}
    \subfigure[The Computation Overhead of the agent]{
        \label{fig:runtime_agent_model1}
        \includegraphics[width=55mm,height=40mm]{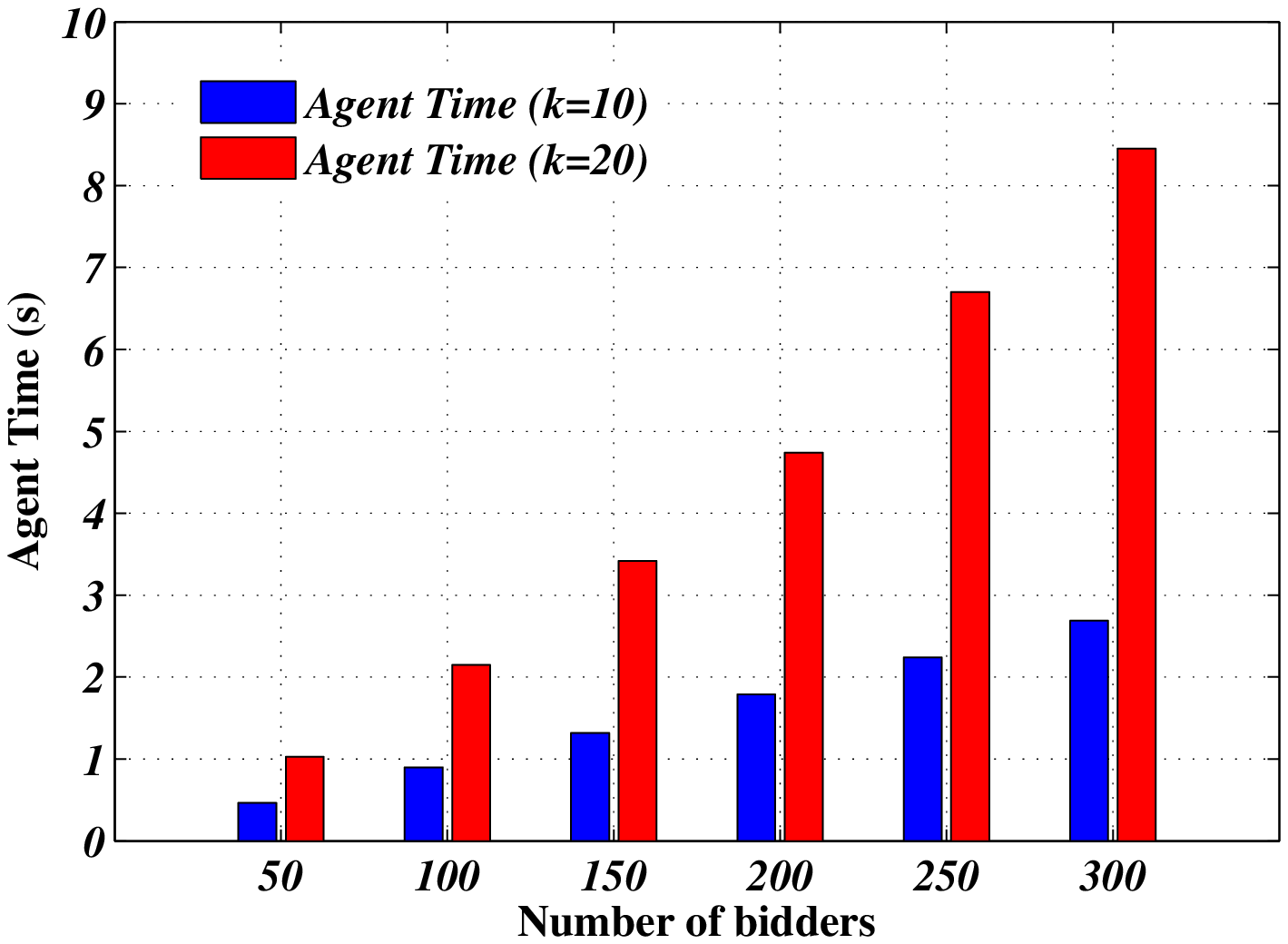}
    }
    \hspace{1mm}
    \subfigure[The Computation Overhead of the Auctioneer]{
        \label{fig:runtime_auctioneer_model1}
        \includegraphics[width=55mm,height=40mm]{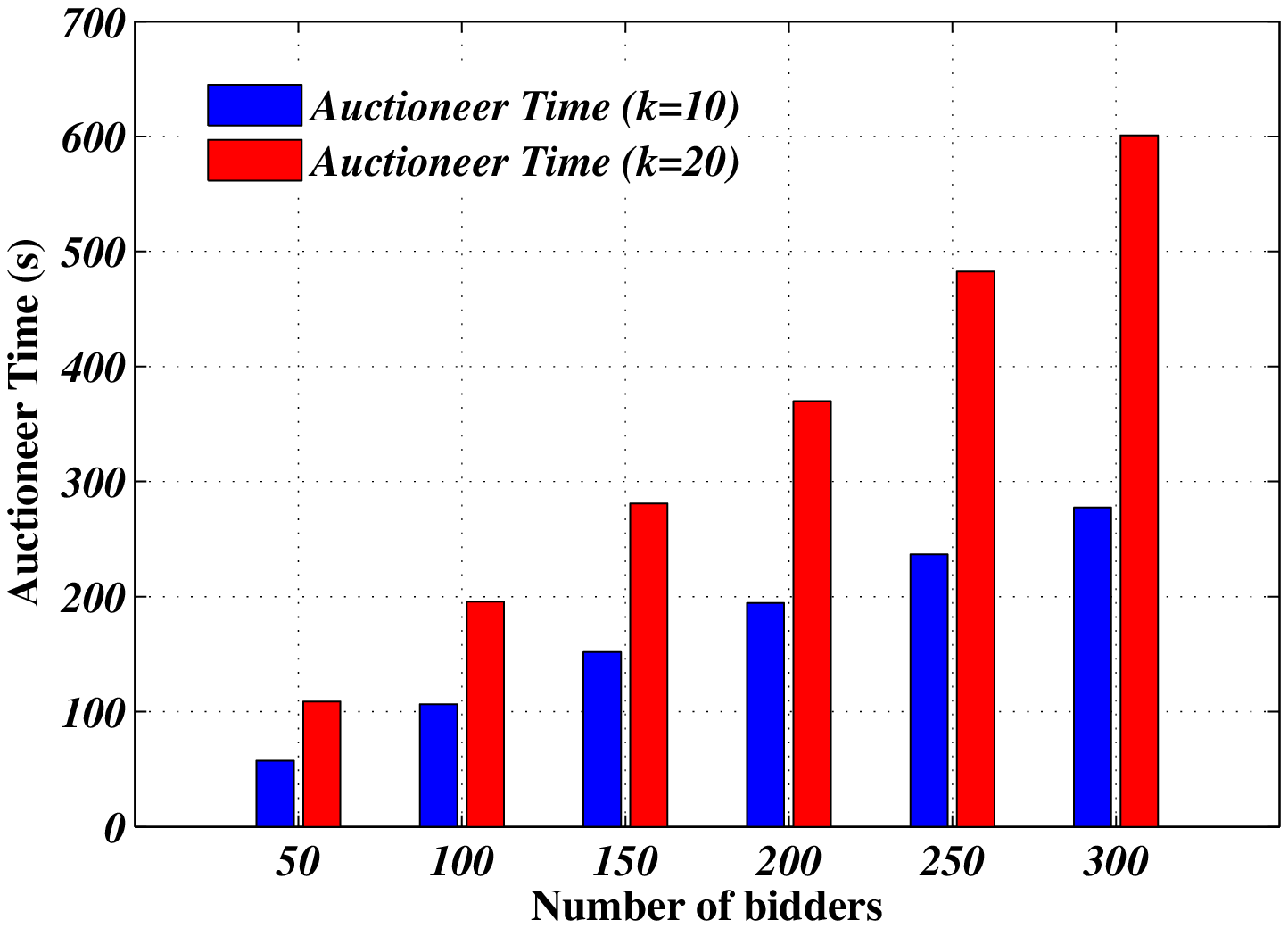}
    }
    \caption{The performance of PPS under SUA model. Here all the bidders are uniformly distributed in a $100 \times 100$ square area.}
    \label{fig:model1}
\end{figure*}

\begin{figure*}[!tbp]
    \centering
    \subfigure[The PPS-MUA and PPS-EMUA Performance of Social Efficiency Ratio under MUA model]{
        \label{fig:MUA_socialefficiency}
        \includegraphics[width=55mm,height=40mm]{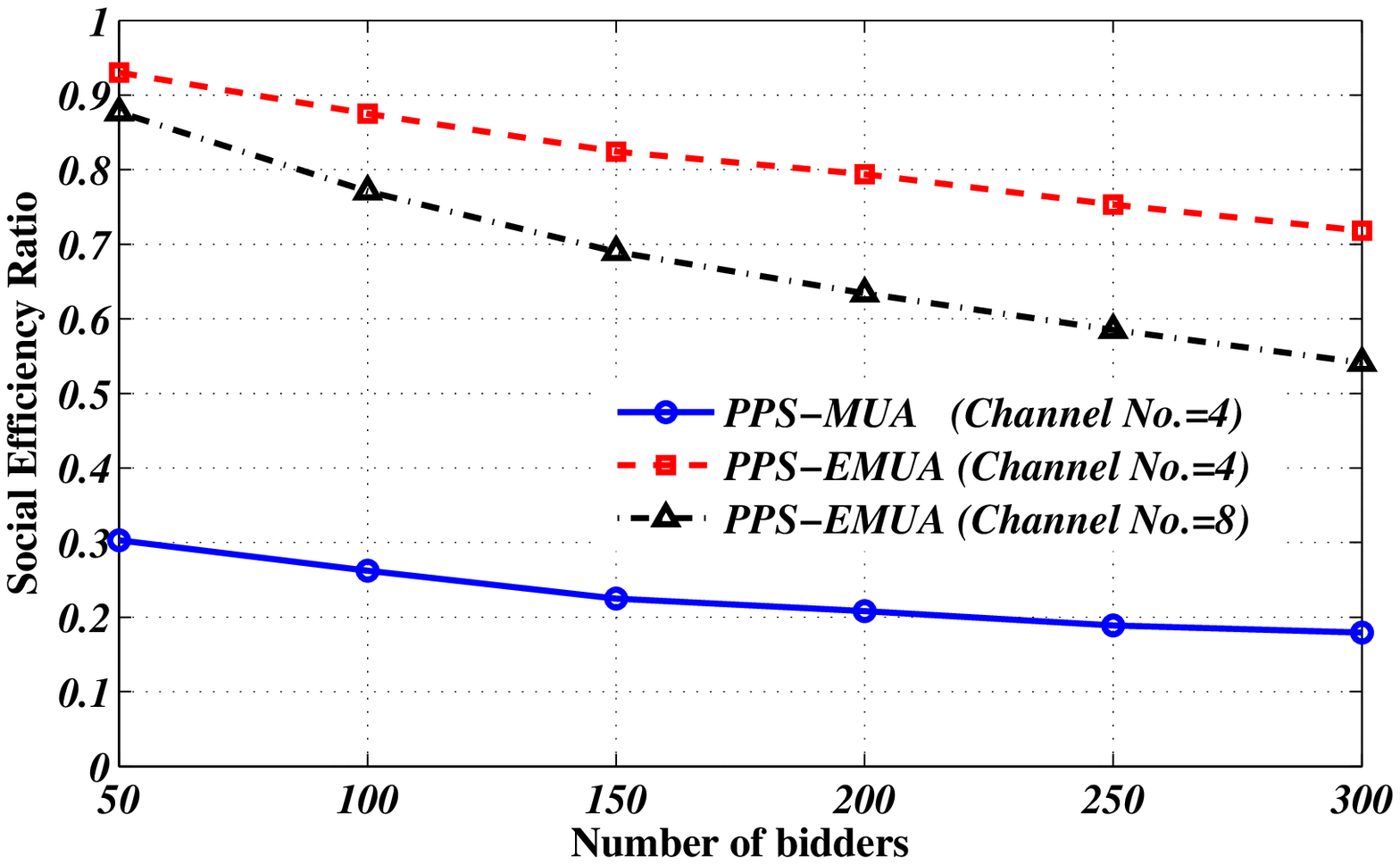}
    }
    \hspace{1mm}
    \subfigure[The Computation Overhead of the Agent]{
        \label{fig:model2_runtime_agent}
        \includegraphics[width=55mm,height=40mm]{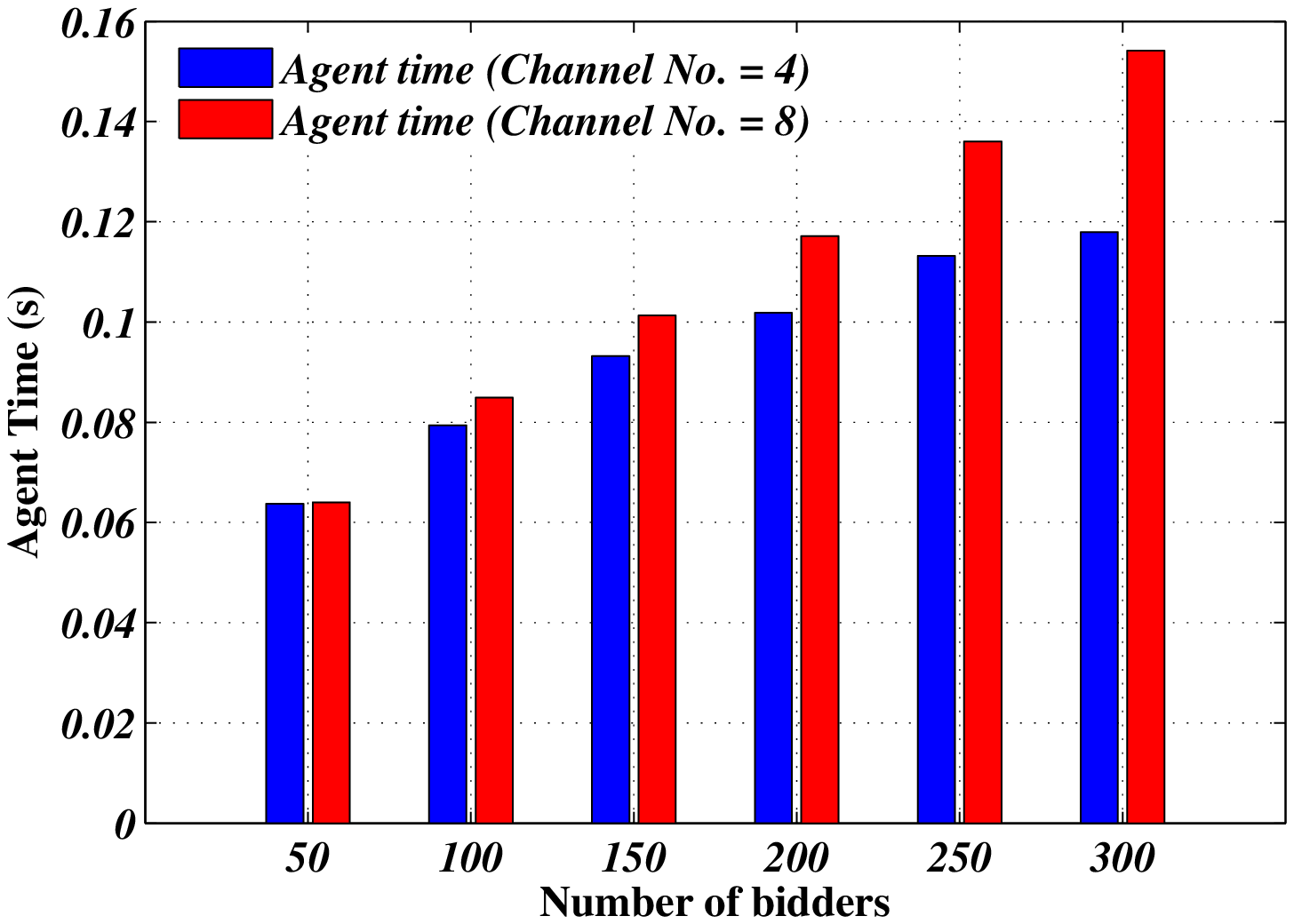}
    }
    \hspace{1mm}
    \subfigure[The Computation Overhead of the Auctioneer]{
        \label{fig:model2_auctioneer_time}
        \includegraphics[width=55mm,height=40mm]{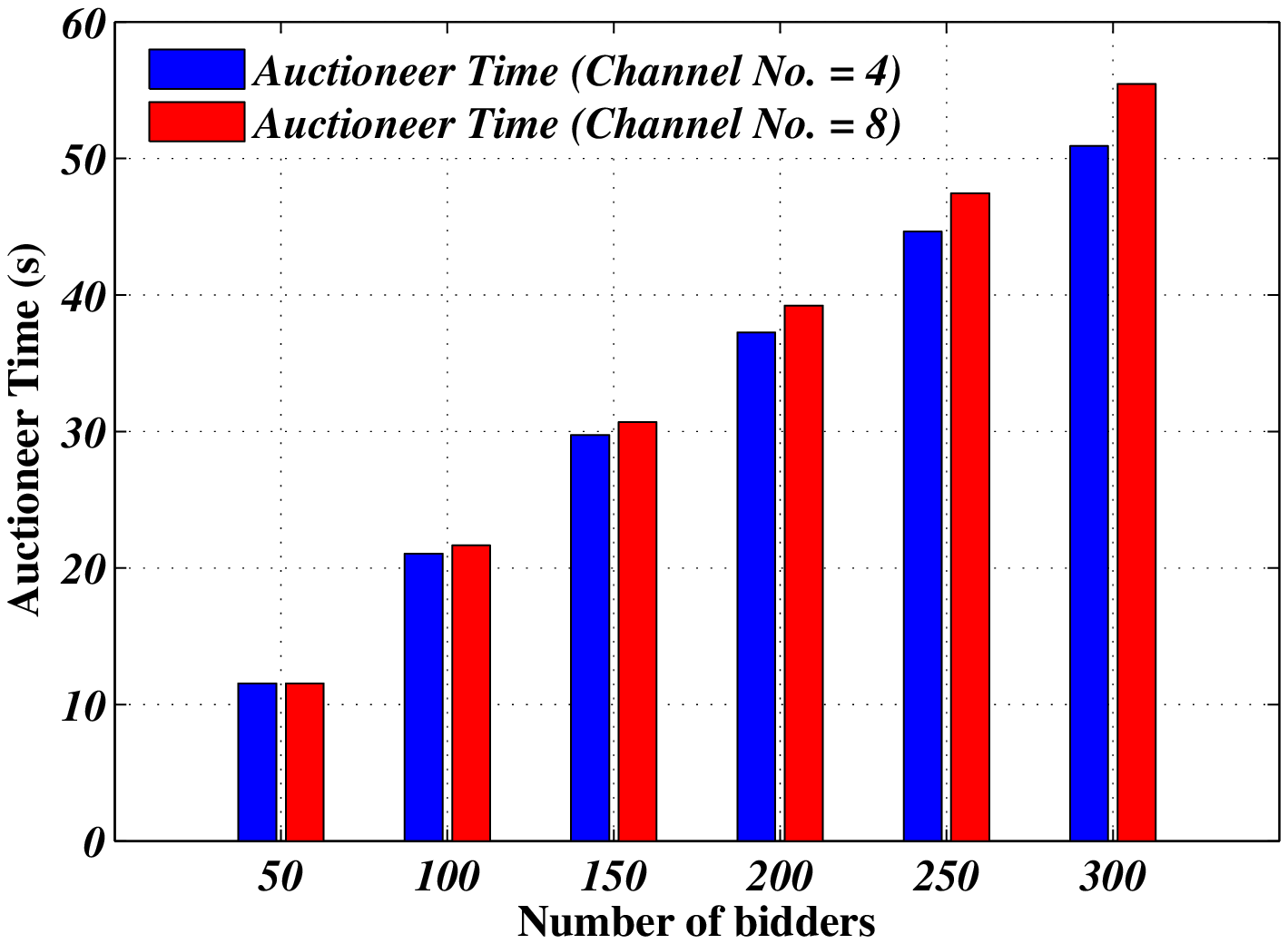}
    }
    \caption{The performance of PPS under MUA model. Here all the bidders are uniformly distributed in a $100 \times 100$ square area, and the channel demand of each bidder is randomly generated from 1 to 4.}
    \label{fig:model2}
\end{figure*}

\subsection{Performance of the PPS}

In this section, we mainly focus on the performance of
social efficiency ratio, auction computation overhead, and
communication overhead under different simulation settings.

We first study the \emph{social efficiency ratio} of our mechanisms under SUA model and
MUA model respectively.
From Fig. \ref{fig:model1social} and Fig. \ref{fig:MUA_socialefficiency}, obviously,
the social efficiency ratio decreases when the number of bidders increases.
This is because the increasing number of bidders will incur a more fierce degree of competition.
Therefore, the social efficiency ratio decreases slightly with the increasing number of bidders in both auction models.
Fig. \ref{fig:model1social} also shows that the social efficiency ratio increases when
$k$ increases, where $k$ is the size of a subdivided grid.
From the theoretical analysis, we can learn that when $k$
increases, less unit-disk defined by bidders' requests are thrown away by using the shifting method.
Thus, the social efficiency ratio increases with the increase of parameter $k$.
Of course, the performances of our proposed PPS-SUA is always better
than the theoretical bound in performance
analysis.
Specifically, Fig. \ref{fig:MUA_socialefficiency} examines the social efficiency
ratio achieved by PPS-MUA and extended version of PPS-MUA (\emph{a.k.a PPS-EMUA}).
We can observe that the ratio of PPS-EMUA performs much better than PPS-MUA when the available channels in spectrum market is fixed to 4.
We can also observe that the PPS-EMUA greatly improves the performance in Fig. \ref{fig:MUA_socialefficiency}.
This is because the PPS-EMUA adopts
a greedy-like allocation mechanism to allocate channels to the potential bidders who lose in PPS-MUA.

Then we study the computation overhead of the proposed mechanisms that were depicted in Fig. \ref{fig:runtime_agent_model1} and Fig. \ref{fig:model2_runtime_agent}.
It is obvious that the computation overhead of the agent change greatly as the number of bidders and $k$ in PPS-SUA.
We can also find that the computation overhead of the agent is increased with the number of bidders,
and affected by the changing of the number of channels slightly in Fig. \ref{fig:model2_runtime_agent}.

Similar to the agent computation overhead, Fig. \ref{fig:runtime_auctioneer_model1} and Fig. \ref{fig:model2_auctioneer_time} plot computation overhead of the auctioneer.
We find that the cost time of auctioneer is much larger than that of the agent,
this is because that the decryption operation cost much more time than the homomorphic operations and auctioneer is responsible for all the decryption operations.

\begin{table}
    \centering
    \caption{Communication Overhead under SUA model (KB)}
    \label{tab:commsua}
    \begin{tabular}{|c|c|c|c|c|c|c|}
    \hline
    \multirow{2}{*}{k} & \multicolumn{6}{|c|}{Number of bidders}\\
    \cline{2-7}
    & 50 & 100 & 150 & 200 & 250 & 300\\
    \hline
    k=10 & 124 & 233 & 333 & 428 & 521 & 611 \\
    \hline
    k=20 & 231 & 416 & 601 & 799 & 1026 & 1273 \\
    \hline
    k=30 & 327 & 603 & 926 & 1312 & 1779 & 2619 \\
    \hline
    \end{tabular}
\end{table}

\begin{table}
    \centering
    \caption{Communication Overhead under MUA model (KB)}
    \label{tab:commmua}
    \begin{tabular}{|c|c|c|c|c|c|c|}
    \hline
    \multirow{2}{*}{Channel Number} & \multicolumn{6}{|c|}{Number of bidders}\\
    \cline{2-7}
    & 50 & 100 & 150 & 200 & 250 & 300\\
    \hline
    4 & 33.5 & 61.9 & 87.5 & 110.8 & 132.2 & 153.0 \\
    \hline
    8 & 34.2 & 63.7 & 90.7 & 117.2 & 140.6 & 164.1 \\
    \hline
    12 & 34.4 & 63.8 & 91.1 & 116.7 & 142.0 & 165.1 \\
    \hline
    \end{tabular}
\end{table}

Table \ref{tab:commsua} and Table \ref{tab:commmua} show the overall communication
overhead induced under SUA and MUA respectively.
We can easily get that the communication overhead is increased with the increment of
number of bidders and $k$ in Table \ref{tab:commsua}.
In Table \ref{tab:commmua}, the total number of channels also plays an important role
in the cost of communication overhead. Anyway, the overheads of the proposed
PPS mechanism are appropriate to be applied in real auction systems.

\section{Literature Reviews}\label{sec:review}
Auctions have been widely used in the scope of
dynamic spectrum allocation. 
Large amount of studies are proposed aiming at designing economical robust
spectrum auction mechanisms (\eg \cite{al2011truthful,zhou2008ebay,zhou2009trust,wu2011small,xu2011efficient,combinatorialinfocom2013libaochun,
dong2012combinatorial,wang2010toda,xu2011truthful,huang2013mobihoc,zhu2012truthful,
deek2011preempt,wangdistrict,gopinathan2010strategyproof,wufan2013mobihoc}).
Each of these approaches has its own optimization goal.
For instance, ~\cite{gopinathan2010strategyproof,dong2012combinatorial,
huang2013mobihoc,zhu2012truthful,xu2011efficient} aim at maximizing the social efficiency
while ensuring strategyproofness in an auction design, and \cite{al2011truthful} aims at achieving the optimal revenue
for the primary user. In ~\cite{deek2011preempt,wang2010toda,xu2011truthful}, the authors consider the truthful online spectrum auction design.
Wu \etal \cite{wu2011small} and Xu \etal \cite{xu2011truthful,xu2011efficient} proposed spectrum auction mechanisms for multi-channel
wireless networks.
Zhou \etal \cite{zhou2009trust} and Wang \etal \cite{wangdistrict} solve the spectrum allocation in a double auction framework.
Unfortunately, none of the above studies addresses the privacy preserving issue in
the auction design.

Although many privacy preserving mechanisms have been proposed in
mechanism design ~\cite{sui2011efficiency,damgaard2007efficient,liqinghua2013percom},
these methods cannot be directly applied in spectrum auction design due to various reasons
 (such as spectrum spatial reuse, computationally complexity).
Recent years, many research efforts focus  on privacy preserving study in auction design ~\cite{naor1999privacy,cachin1999efficient}.
Huang \etal \cite{wufan2013spring} first propose a strategyproof spectrum auction with consideration of
privacy preserving, and Pan \etal \cite{pan2011purging} provide a secure spectrum auction to
prevent the frauds of the insincere auctioneer.
Unfortunately,  none of the existing solutions with privacy preserving provides any
performance guarantee, such as maximizing the social efficiency  which is often NP-hard.
Our mechanisms rely on privacy preserving comparison and polynomial evaluations \cite{taeho-info13}, which
 is extensively studied topic in   secure multi-party computation~\cite{yao1982protocols,damgard2008homomorphic,aggarwal2004secure,du2001secure}.

\section{Conclusion}\label{sec:conclusion}
In this paper, we  focused on  designing strategyproof auction mechanisms which
maximize the social efficiency without leaking any true bid value of
bidders, and proposed 
a framework of PPS for solving this issue. 
We   designed privacy-preserving strategyproof
auction mechanisms with approximation factors of $(1+\epsilon)$ and $32$ separately for SUA and MUA.
Our evaluation results
demonstrated that both PPS-SUA and PPS-EMUA achieve good performance on social efficiency, while
inducing only a small amount of computation and communication overhead.
A future work is to design robust privacy-preserving strategyproof
auction mechanisms without inexplicitly requiring the location of
bidders.
Another future work is to design privacy-preserving auction mechanisms
by removing the dependency of third-party agent.

\section*{Acknowledgement}
The research of authors is partially supported by the National Grand
Fundamental Research 973 Program of China (No.2011CB302905,
No.2011CB302705), National Natural Science Foundation of China (NSFC)
under Grant No. 61202028, No. 61170216, No. 61228202, and NSF
CNS-0832120, NSF CNS-1035894, NSF ECCS-1247944. Specialized Research
Fund for the Doctoral Program of Higher Education (SRFDP) under Grant
No. 20123201120010.

\ifCLASSOPTIONcaptionsoff
  \newpage
\fi

{
\bibliography{auction}
}

\end{document}